\documentclass{article}

\usepackage[english]{babel}

\usepackage{arxiv}

\usepackage{amssymb}
\usepackage{times}
\usepackage{cite}

\usepackage{amsthm}
\usepackage{amsxtra}
\usepackage{amsmath}

\usepackage{verbatim}
\usepackage{epsfig}
\usepackage{setspace}
\usepackage{caption}
\usepackage{breqn}
\usepackage{graphicx}
\usepackage{tikz}
\usetikzlibrary{positioning}
\usepackage[colorlinks=true, allcolors=blue]{hyperref}

\newtheorem{theorem}{Theorem}
\newtheorem{definition}{Definition}
\newtheorem{lemma}[theorem]{Lemma}
\newtheorem{proposition}[theorem]{Proposition}
\newtheorem{corollary}[theorem]{Corollary}

\newtheorem{problem}{Problem}
\newtheorem{remark}{Remark}

\title{On Characterization of Finite Geometric Distributive Lattices\thanks{This work is fully supported by the NTRO project at R. C. Bose Center for Cryptology and Security at Indian Statistical Institute, Kolkata.}}
\author{ \href{https://orcid.org/0000-0002-9110-7617}{\hspace{1mm}Pranab Basu}\\
	R. C. Bose Center for Cryptology and Security\\
	Indian Statistical Institute\\
	Kolkata 700108 \\
	\texttt{pranabbasu@alum.iisc.ac.in} \\}
\date{}

\renewcommand{\headeright}{}
\renewcommand{\undertitle}{}
\renewcommand{\shorttitle}{On Characterization of Finite Geometric Distributive Lattices}

\begin{document}
\maketitle

\begin{abstract}
A Lattice is a partially ordered set where both least upper bound and greatest lower bound of any pair of elements are unique and exist within the set. K\"{o}tter and Kschischang proved that codes in the linear lattice can be used for error and erasure-correction in random networks. Codes in the linear lattice have previously been shown to be special cases of codes in modular lattices. Two well known classifications of semimodular lattices are geometric and distributive lattices. Most of the frequently used coding spaces are examples of either or both. We have identified the unique criterion which makes a geometric lattice distributive, thus characterizing all finite geometric distributive lattices. Our characterization helps to prove a conjecture regarding the maximum size of a distributive sublattice of a finite geometric lattice and identify the maximal case. The Whitney numbers of the class of geometric distributive lattices are also calculated. We present a few other applications of this unique characterization to derive certain results regarding linearity and complements in the linear lattice.
\end{abstract}

\setlength{\parindent}{10ex}
\section{Introduction}
\label{S1}
Let $\mathbb{F}_q^n$ be the $n$-dimensional vector space over $\mathbb{F}_q$, the unique finite field with $q$ elements; $q$ is necessarily a prime power. The set of all subspaces of $\mathbb{F}_q^n$ is the \emph{projective space}\footnote{This terminology is not standard. In other branches of mathematics the term projective space defines the collection of all lines passing through the origin of a vector space.} $\mathbb{P}_q(n)$ which can be formally defined as
\begin{equation*}
	\mathbb{P}_q(n) := \{V : V \le \mathbb{F}_q^n\},
\end{equation*}
where $\le$ signifies the usual vector space inclusion. The collection of all subspaces in $\mathbb{P}_q(n)$ with a fixed dimension $k$ is called the \emph{Grassmannian} of dimension $k$ for all $0 \le k \le n$, and is denoted as $\mathbb{G}_q(n, k)$. In terms of notation, $\mathbb{G}_q(n, k) := \{V: V \le \mathbb{F}_q^n, \dim V = k\}$. Clearly, $\mathbb{P}_q(n) = \bigcup\limits_{k=0}^{n} \mathbb{G}_q(n, k)$. The \emph{subspace distance} between two subspaces $X$ and $Y$ in $\mathbb{P}_q(n)$ is defined as
\begin{equation*}
	d_S(X, Y) := \dim (X+Y) - \dim (X \cap Y),
\end{equation*}
where $X+Y$ denotes the smallest subspace containing both $X$ and $Y$. It was proved in \cite{KK, AAK} that the projective space $\mathbb{P}_q(n)$ is a metric space under the action of the subspace distance metric. A \emph{code} in the projective space $\mathbb{P}_q(n)$ is a subset of $\mathbb{P}_q(n)$.

Codes in projective spaces have recently gained attention since they were proved to be useful for error and erasure-correction in \emph{random network coding} \cite{KK}. An $(n, M, d)$ code in $\mathbb{P}_q(n)$ is a collection of $M$ number of subspaces of $\mathbb{F}_q^n$ such that the minimum subspace distance between any two of them is $d$. K\"{o}tter and Kschischang showed that an $(n, M, d)$ code can correct any combination of $t$ errors and $\rho$ erasures during the communication of packets through a volatile network as long as $2(t+\rho) < d$ \cite{KK}. Subsequently codes in $\mathbb{P}_q(n)$ were studied extensively \cite{EV, HKK, BP, GR}. Such codes are also referred to as \emph{subspace codes}.

However, designing and studying code structures in $\mathbb{P}_q(n)$ is considered relatively trickier than study of classical error-correction in the \emph{Hamming space} $\mathbb{F}_q^n$. This is because unlike in $\mathbb{F}_q^n$, the volume of a \emph{sphere} is not independent of the choice of its \emph{center} in the projective space $\mathbb{P}_q(n)$. Thus, standard geometric intuitions often do not hold in $\mathbb{P}_q(n)$. In other words, $\mathbb{F}_q^n$ is \emph{distance-regular} while $\mathbb{P}_q(n)$ is not. This implies that a different framework is required to study codes in projective spaces than the approach taken for block codes in classical error-correction, e.g. in \cite{MS}. The problem of lack of distance-regularity in $\mathbb{P}_q(n)$ is, however, tackled to some extent by considering codewords of a fixed dimension. Such class of subspace codes are commonly known as \emph{constant dimension codes}. A constant dimension code in $\mathbb{P}_q(n)$ is a subset of $\mathbb{G}_q(n, k)$ for some $0 \le k \le n$. The fact that $\mathbb{G}_q(n, k)$ is distance-regular is exploited to construct various classes of constant dimension codes \cite{SE, SE2, GY, XF, TMBR, KoK, ER}.

A lattice framework for studying both binary block codes and subspace codes was discussed in \cite{MB}. The authors of \cite{BEV} established that codes in projective spaces are $q$-analogs of binary block codes defined within Hamming spaces using a framework of lattices. A \emph{lattice} is a partially ordered set wherein both least upper bound and greatest lower bound of any pair of elements exist within the set and are unique. We denote the set of all subsets of the canonical $n$-set $[n] := \{1, \ldots, n\}$ as $\mathcal{P}(n)$, also known as the \emph{power set} of $[n]$. The lattices corresponding to the block codes in $\mathbb{F}_2^n$ and the subspace codes in $\mathbb{P}_q(n)$ are the \emph{power set lattice} $(\mathcal{P}(n), \cup, \cap, \subseteq)$ and the \emph{linear lattice} $(\mathbb{P}_q(n), +, \cap, \le)$, respectively. Here the notation $\subseteq$ represents set inclusion.

A lattice is called \emph{modular} if the modularity condition holds for all elements in it (see Def.~\ref{LD5}). Many of the well-known coding spaces including both $\mathbb{F}_q^n$ and $\mathbb{P}_q(n)$ are examples of a modular lattice. This motivated the work of Kendziorra and Schmidt where they generalized the model of subspace codes introduced in \cite{KK} to codes in modular lattices \cite{KS}. There are two significant types of semimodular lattices, viz. \emph{geometric} lattices and \emph{distributive} lattices that have inspired a rich variety of literature, e.g. \cite{B, RS}. While $\mathbb{F}_2^n$ is a geometric distributive lattice, $\mathbb{P}_q(n)$ is an example of a geometric lattice which is modular but non-distributive. There have been quite a few attempts to characterize distributive lattices, such as in \cite{LS, S}. However, no known characterization of geometric distributive lattices exists to the best of our knowledge.

The notion of ``linearity" and ``complements" in $\mathbb{P}_q(n)$ are not as straightforward as they are in the Hamming space $\mathbb{F}_2^n$. This is owed to the fact that $\mathbb{F}_2^n$ is a vector space with respect to the bitwise XOR-operation whereas $\mathbb{P}_q(n)$ or $\mathbb{G}_q(n, k)$ are not vector spaces with respect to the usual vector space addition. Therefore, the subspace distance metric is not \emph{translation invariant} over $\mathbb{P}_q(n)$ or $\mathbb{G}_q(n, k)$. Braun et al. addressed this problem in \cite{BEV} and defined linearity and complements in subsets of $\mathbb{P}_q(n)$ by elucidating key features from the equivalent notion in $\mathbb{F}_2^n$.

The maximum size of a linear code in $\mathbb{P}_2(n)$ was conjectured to be $2^n$ by Braun et al. \cite{BEV}. A particular case of this problem was proved by Pai and Rajan where the ambient space $\mathbb{F}_2^n$ is included as a codeword \cite{PS}. The maximal code achieving the upper bound was identified as a \emph{code derived from a fixed basis}. The authors of \cite{PS} observed that such a code is basically embedding of a distributive lattice into the linear lattice, which is geometric. This motivated them to conjecture a generalized statement which can already be found in literature, e.g. in \cite[Ch.~IX, Sec.~4, Ex.~1]{B}.
\begin{problem}
	\label{Pr1}
	The size of the largest distributive sublattice of a gemetric lattice of height $n$ must be $2^n$.
\end{problem}
Lattice-theoretic connection of other classes of linear codes in $\mathbb{P}_q(n)$ was investigated thoroughly in \cite{BK}. The findings of \cite{BK} include the discovery that the only class of linear subspace codes that have a sublattice structure of the corresponding linear lattice must be geometric distributive. Thus it is an interesting problem to find out a unique characterization of geometric distributive lattices should it exist.

In this paper, we determine the unique criterion for a geometric lattice to be distributive. We in fact prove a more generalized version of this statement. This helps us to bring out the unique characterization of class of geometric distributive lattices. We then use this characterization to solve a few problems involving linear codes and complements in $\mathbb{P}_q(n)$. Problem~\ref{Pr1} is also solved by applying the said characterization.

The rest of the paper is organized as follows. In Section~\ref{S2} we give a few requisite definitions concerning lattices and formally define linear codes and complements in the projective space $\mathbb{P}_q(n)$. Section~\ref{UAL} concerns with the study of uniquely atomistic lattices; in particular we show that any such finite lattice is modular. The \emph{unique-decomposition theorem} that gives the unique characterization of finite geometric distributive lattices is derived in Section~\ref{S3} after proving a sequence of results regarding modular lattices and distributive lattices. Section~\ref{S4} is attributed to various applications of the unique-decomposition theorem in lattice theory that include determining the maximum size of a distributive sublattice of a finite geometric lattice and counting the \emph{Whitney numbers} of a geometric distributive lattice. In particular, we consider a few problems about linearity and complements in the linear lattice. An important finding is that any distributive sublattice of $\mathbb{P}_q(n)$ can be used to construct a linear code closed under intersection. Concluding remarks and interesting open problems are listed in Section~\ref{S5}.
\paragraph{Notation.}
$\mathbb{F}_q^n$ represents the unique vector space of dimension $n$ over $\mathbb{F}_q$. The set of all subspaces of $\mathbb{F}_q^n$ is denoted as $\mathbb{P}_q(n)$. The usual vector space sum of two disjoint subspaces $X$ and $Y$, called the \emph{direct sum} of $X$ and $Y$, is written as $X \oplus Y$. For any subset $\mathcal{U} \subseteq \mathbb{P}_q(n)$, the collection of all $i$-dimensional members of $\mathcal{U}$ will be denoted as $\mathcal{U}_i$; $\mathcal{U}_i := \{X: X \in \mathcal{U}, \dim X = i\}$. The notation $\langle \mathcal{S}\rangle$ for any subset $\mathcal{S}$ of vectors in $\mathbb{F}_q^n$ will denote the linear span of all the vectors in $\mathcal{S}$. $\triangle$ denotes the \emph{symmetric difference} operator, which can be defined for two sets $S$ and $T$ as
\begin{equation*}
	S \triangle T := (S \cup T) \backslash (S \cap T).
\end{equation*}
\section{Preliminaries}
\label{S2}
\subsection{An Overview of Lattices}
We will go through some standard definitions and results concerning lattices that can be found in the existing literature, e.g. in \cite{B}.
\begin{definition}
	\label{LD1}
	For a set $P$, the pair $(P, \preceq)$ is called a \emph{poset} if there exists a binary relation $\preceq$ on $P$, called the \emph{order relation}, that satisfies the following for all $x, y, z \in P$:
	\begin{itemize}
		\item[(i)] (Reflexivity) $x \preceq x$;
		\item[(ii)] (Antisymmetry) If $x \preceq y$ and $y \preceq x$, then $x = y$; and
		\item[(iii)] (Transitivity) If $x \preceq y$ and $y \preceq z$, then $x \preceq z$.
	\end{itemize}
	The \emph{dual} of a poset $P$ is the poset $P^{*}$ defined on the same set as $P$ such that $y \preceq x$ in $P^{*}$ if and only if $x \preceq y$ in $P$.
\end{definition}
The notation $x \preceq y$ is read as ``$x$ is less than $y$'' or ``$x$ is contained in $y$''. If $x \preceq y$ such that $x \ne y$, then we write $x \prec y$. In the sequel a poset $(P, \preceq)$ will be denoted as $P$ when the order relation $\prec$ is obvious from the context.
\begin{definition}
	\label{LD2}
	An upper bound (lower bound) of a subset $S$ of a poset $P$ is an element $p \in P$ containing (contained in) every $s \in S$. A least upper bound (greatest lower bound) of $S \subseteq P$ is an element of $P$ contained in (containing) every upper bound (lower bound) of $S$.
\end{definition}
A least upper bound or a greatest lower bound of a poset, should it exist, is unique according to the antisymmetry property of the order relation $\preceq$. The least upper bound and the greatest lower bound of a poset $P$ are denoted as $\sup P$ and $\inf P$, respectively.
\begin{definition}
	\label{LD3}
	A \emph{lattice} $(L, \vee, \wedge)$ is a poset $L$ such that $\sup \{x, y\}$ and $\inf \{x, y\}$ exist for all $x, y \in L$. The notation for the $\sup \{x, y\}$ and the $\inf \{x, y\}$ are $x \vee y$ (``$x$ \emph{join} $y$'') and $x \wedge y$ (``$x$ \emph{meet} $y$''), respectively.
\end{definition}
Once again, a lattice $(L, \vee, \wedge)$ will be denoted as $L$ whenever the join $\vee$ and meet $\wedge$ operations are obvious from the context. In this work we will consider only finite lattices, i.e. when the underlying poset is finite. The unique greatest element and the unique least element of a lattice will be denoted as $I$ and $O$, respectively, unless specified otherwise.
\begin{definition}
	\label{LD4}
	A \emph{sublattice} of a lattice $L$ is a subset $S \subseteq L$ such that $x \vee y, x \wedge y \in S$ for all $x, y \in S$.
\end{definition}
The \emph{Hasse diagram} of a finite poset completely describes the order relations of that poset. If $x \prec y$ in the poset $P$ such that there exists no $z \in P$ satisfying $x \prec z \prec y$, then $y$ is said to \emph{cover} $x$; we denote this as $x \lessdot y$. In the Hasse diagram of a lattice, two elements are joined if and only if one of them covers the other; $y$ is written above $x$ if $y$ covers $x$. Hence, $x \prec y$ if and only if there exists a path from $x$ moving up to $y$.

The power set $\mathcal{P}(m)$ of a finite set $[m]$ and the projective space $\mathbb{P}_q(n)$ are examples of lattices. The Hasse diagram associated with the lattice of $(\mathbb{P}_2(2), +, \cap)$ is shown here (Fig.~\ref{F1}). This particular lattice is known as $M_3$.
\begin{figure}[t]
	\centering
	\begin{tikzpicture}[scale=0.6]
		\node (A1) at (0,3) {$\mathbb{F}_2^2$};
		\node (A2) at (-3,0) {$\langle \{(0, 1)\}\rangle$};
		\node (A3) at (0,0) {$\langle \{(1, 0)\}\rangle$};
		\node (A4) at (3,0) {$\langle \{(1, 1)\}\rangle$};
		\node (A5) at (0,-3) {$\{0\}$};
		\draw (A5) -- (A2) -- (A1) -- (A3) -- (A5) -- (A4) -- (A1);
	\end{tikzpicture}
	\caption{$M_3$ lattice representing $(\mathbb{P}_2(2), +, \cap)$}
	\label{F1}
\end{figure}
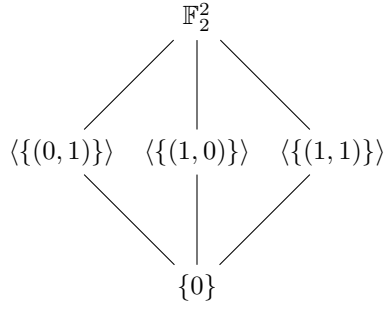
\begin{definition}
	\label{LD5}
	A finite lattice $(L, \vee, \wedge)$ is \emph{semimodular} if the following holds for all $x, y \in L$:
	\begin{equation*}
		x \wedge y \lessdot x, y \quad \Rightarrow \quad x, y \lessdot x \vee y.
	\end{equation*}
	A lattice is \emph{modular} if both the lattice and its dual are semimodular. It can be proved that a finite lattice $L$ is modular if for any $x, y, z \in L$ the following holds:
	\begin{equation*}
		x \preceq z \Rightarrow x \vee (y \wedge z) = (x \vee y) \wedge z.
	\end{equation*}
\end{definition}
The smallest finite lattice that is non-modular is called $N_5$ (Fig.~\ref{F0}). $N_5$ plays a crucial role in characterizing modular lattices as we will see later.

The elements of a lattice which cover the least element of the lattice are known as \emph{atoms}. A lattice with a least element is \emph{atomic} if for every non-zero element $a$ there exists an atom $p$ such that $p \preceq a$. An atomic lattice is called \emph{atomistic} if any element is a join of atoms. A lattice that is \emph{uniquely} atomistic is defined in the following way.
\begin{definition}
	\label{LDA}
	A lattice is \emph{uniquely atomistic} if each element therein is uniquely expressible as join of its atoms. If $L$ is a uniquely atomistic lattice with $\{x_1, \ldots, x_m\}$ as the set of all atoms in $L$ then for any $x \in L$ there exists a unique subset $S_x \subseteq [m]$ such that $x = \bigvee\limits_{i \in S_x} x_i$. We denote this relation as $x = \sup S_x$ when the choice of $m$ is clear from the context.
\end{definition}
Atoms play an important role in defining geometric lattices.
\begin{definition}
	\label{LD6}
	A finite lattice that is both semimodular and atomistic is called \emph{geometric}.
\end{definition}
Both the linear lattice $(\mathbb{P}_q(n), +, \cap)$ and the power set lattice $(\mathcal{P}(m), \cup, \cap)$ are examples of a geometric lattice. From Section~\ref{S4} onwards all lattices considered will be geometric. The variety of modular lattices that will play a key role in this work are the distributive lattices which are defined next.
\begin{definition}
	\label{LD7}
	A lattice $L$ is \emph{distributive} if the following two equivalent conditions hold for any $x, y, z \in L$:
	\begin{eqnarray}
		x \vee (y \wedge z) &=& (x \vee y) \wedge (x \vee z); \nonumber \\
		x \wedge (y \vee z) &=& (x \wedge y) \vee (x \wedge z). \nonumber
	\end{eqnarray}
\end{definition}
The power set lattice $\mathcal{P}(m)$ is an example of a distributive lattice. However, the linear lattice $\mathbb{P}_q(n)$ is modular but not distributive. In general, any distributive lattice is modular, and the $M_3$ lattice is pivotal in characterizing modular non-distributive lattices. Similarly a modular lattice can be defined by non-inclusion of the lattice $N_5$. The following theorem is due to Dedekind and Birkhoff.
\begin{theorem}(\cite{G}, Page~59)
	\label{TL1}
	A lattice is modular if and only if it does not contain a sublattice isomorphic to $N_5$. A modular lattice is non-distributive if and only if it contains a sublattice isomorphic to $M_3$.
\end{theorem}
\begin{definition}
	\label{LD8}
	A real valued function $v: L \rightarrow \mathbb{R}$ on a lattice $L$ is called a \emph{positive isotone valuation} if the following conditions hold for all $x, y \in L$:
	\begin{itemize}
		\item[(i)] (Valuation) $v(x \vee y) + v(x \wedge y) = v(x) + v(y)$; 
		\item[(ii)] (Isotone) $x \preceq y \Rightarrow v(x) \le v(y)$;
		\item[(iii)] (Positive) $x \prec y \Rightarrow v(x) < v(y)$.
	\end{itemize}
\end{definition}
The distance function induced by an isotone valuation $v$ is defined as $d_v(x, y) := v(x \vee y) - v(x \wedge y)$.
\begin{theorem}(\cite{B})
	\label{TL2}
	For an isotone valuation $v$ defined on a lattice $L$, the function $d_v(x, y) := v(x \vee y) - v(x \wedge y)$ is a metric if and only if $v$ is positive.
\end{theorem}
A totally ordered subset of a lattice is called a \emph{chain}. Given two elements $x$ and $y$ in a lattice $L$, a chain of $L$ between $x$ and $y$ is a chain $\{x_1, \ldots, x_l\}$ such that $x = x_0 \prec x_1 \prec \cdots \prec x_l = y$. The \emph{length} of this chain is $l$. The \emph{height} of an element $x \in L$ is the maximum length of all chains between $O$ and $x$, and is denoted by $h_L(x)$. We often use the notation $h(x)$ when $L$ is obvious from the context. The \emph{height of the lattice} $L$ is the number $h_L(I)$ where $I$ is the greatest element in $L$.

For modular lattices, the following is a consequence of Theorem~\ref{TL2}.
\begin{theorem}[Page~41, Theorem~16, \cite{B}]
	\label{TL3}
	If $h$ is the height function defined on a finite modular lattice $L$ then $h$ is a positive isotone valuation and $d_h$ is a metric on $L$.
\end{theorem}
By definition, $h_L(x) = 0$ if and only if $x$ is the least element in $L$; similarly, $h_L(x) = 1$ if and only if $x$ is an atom in $L$.
\begin{definition}
	\label{LD9}
	The total number of elements with a given height $k$ of a lattice $L$ with a height function $h$ defined on $L$ is called the \emph{Whitney number}, denoted as $W_k(L)$. In terms of notation, $W_k(L) := |\{x \in L: h_L(x) = k\}|$.
\end{definition}
\subsection{Complements and Linearity in Projective Spaces}
The notions of complements and linearity in the projective space $\mathbb{P}_q(n)$ are not straightforward as they are in the Hamming space $\mathbb{F}_2^n$. Braun et al. introduced the definition of both in \cite{BEV} by extracting key properties of the same in $\mathbb{F}_2^n$. We begin with a formal definition of the complement mapping in $\mathbb{P}_q(n)$.
\begin{definition}
	\label{D1}
	For any subset $\mathcal{U} \subseteq \mathbb{P}_q(n)$, a function $f: \mathcal{U} \rightarrow \mathcal{U}$ is a \emph{complement} on $\mathcal{U}$ if $f$ satisfies the following conditions:
	\begin{itemize}
		\item[(i)] $X \cap f(X) = \{0\}$ and $X \oplus f(X) = \mathbb{F}_q^n$ for all $X \in \mathcal{U}$;
		\item[(ii)] There exists a unique $f(X) \in \mathcal{U}_{n-k}$ for each $X \in \mathcal{U}_k$ for all $0 \le k \le n$;
		\item[(iii)] $f(f(X)) = X$ for all $X \in \mathcal{U}$; and
		\item[(iv)] $d_S(X, Y) = d_S (f(X), f(Y))$ for all $X, Y \in \mathcal{U}$.
	\end{itemize}
\end{definition}
It was proved before that a complement function does not exist in the entirety of $\mathbb{P}_q(n)$ \cite[Theorem~10]{BEV}. The largest size of a subset of $\mathbb{P}_q(n)$ wherein a complement can be defined still remains an open problem. However, we will tackle that question in Section~\ref{S4} with the additional constraint that a subset has distributive sublattice structure. The following is an upper bound on the number of one-dimensional subspaces in a subset with a complement defined on it.
\begin{proposition}({\cite{BEV}, Proposition~1})
	\label{CP1}
	Suppose there exists a complement on the subset $\mathcal{U} \subseteq \mathbb{P}_2(n)$. Then $|\mathcal{U}_1| \le 2^{n-1}$.
\end{proposition}
We will later investigate the same for distributive sublattices of $\mathbb{P}_q(n)$ for all prime powers $q$.

Braun et al. defined linearity in $\mathbb{P}_2(n)$ by identifying a subset that is a vector space over $\mathbb{F}_2$ with respect to some randomly chosen linear operation such that the corresponding subspace distance metric is translation invariant within the chosen subset. Later this definition was generalized for all prime powers \cite{PS, BK2}.
\begin{definition}
	\label{D2}
	A subset $\mathcal{U} \subseteq \mathbb{P}_q(n)$ is called a \emph{linear code} if $\{0\} \in \mathcal{U}$ and there exists a function $\boxplus: \mathcal{U} \times \mathcal{U} \rightarrow \mathcal{U}$ such that
	\begin{itemize}
		\item[(i)] $(\mathcal{U}, \boxplus)$ is an abelian group;
		\item[(ii)] $X \boxplus \{0\} = X$ for all $X \in \mathcal{U}$;
		\item[(iii)] $X \boxplus X = \{0\}$ for all $X \in \mathcal{U}$; and
		\item[(iv)] $d_S(X, Y) = d_S(X \boxplus W, Y \boxplus W)$ for all $X, Y, W \in \mathcal{U}$.
	\end{itemize}
\end{definition}
The first three conditions stated in the above definition makes any linear code in $\mathbb{P}_q(n)$ a vector space over $\mathbb{F}_2$. It was conjectured in \cite{BEV} that a linear code in $\mathbb{P}_2(n)$ can be as large as $2^n$ at most.

The linear addition of two disjoint codewords in a linear code yields their usual vector space sum.
\begin{lemma}(\cite{BEV}, Lemma~8)
	\label{LL1}
	For two codewords $X$ and $Y$ of a linear code $\mathcal{U} \subseteq \mathbb{P}_q(n)$, $X \boxplus Y = X + Y$ if $X \cap Y = \{0\}$.
\end{lemma}
A linear code is said to be \emph{closed under intersection} if it is closed with respect to vector space intersection: The subspace $X \cap Y$ is a codeword of $\mathcal{U}$ for any two codewords $X$ and $Y$ if $\mathcal{U}$ is a linear code closed under intersection. The following is a method to construct such class of linear codes.
\begin{theorem}(\cite{BK2}, Theorem~7)
	\label{LT1}
	Suppose there exists a linearly independent subset $\mathcal{E} = \{e_1, \ldots, e_r\}$ of $\mathbb{F}_q^n$ over $\mathbb{F}_q$. Let $\{\mathcal{E}_1, \ldots, \mathcal{E}_m\}$ be a partition of $\mathcal{E}$. Define $\mathcal{E}_{\mathcal{I}} := \bigcup\limits_{i \in \mathcal{I}} \mathcal{E}_i$ for any nonempty subset $\mathcal{I} \subseteq [m]$ and $\mathcal{E}_{\phi} := \phi$. The code $\mathcal{U} = \{\langle \mathcal{E}_{\mathcal{I}}\rangle: \mathcal{I} \subseteq [m]\}$ is linear and closed under intersection.
\end{theorem}
A linear code thus constructed is also referred to as a \emph{code derived from a partition of a linearly independent set}. The particular case when $r = m = n$ in Theorem~\ref{LT1} is referred to as a \emph{code derived from a fixed basis}, and was first introduced in \cite{PS}. It is known that any linear code closed under intersection can only be constructed in the way described in Theorem~\ref{LT1} \cite[Theorem~8]{BK2}. The lattice structure of such class of linear codes was studied in \cite{BK}.
\begin{theorem}(\cite{BK}, Theorem~18)
	\label{LT2}
	A linear code in $\mathbb{P}_q(n)$ that is closed under intersection forms a distributive sublattice of the linear lattice $\mathbb{P}_q(n)$.
\end{theorem}
The maximum size of a linear code closed under intersection was investigated in \cite{BK2} and it revealed that the maximal case is unique.
\begin{theorem}(\cite{BK2})
	\label{LT3}
	The maximum size of a linear code closed under intersection in $\mathbb{P}_q(n)$ is $2^n$. The bound is reached if and only if the code is derived from a fixed basis.
\end{theorem}
We will exploit the lattice theoretic connection of linear codes further in Section~\ref{S4} using the unique decomposition of geometric distributive lattices.
\section{Uniquely Atomistic Lattices}
\label{UAL}
We will prove modularity of uniquely atomistic lattices in this section via
a series of results. First a few elementary lemmas follow from definition.
\begin{lemma}
	\label{UAL1}
	If $x = \sup S$ and $y = \sup T$ are two elements of a uniquely atomistic lattice $L$, then $x \vee y = \sup (S \cup T)$.
\end{lemma}
\begin{proof}
	By definition, if the set of all atoms in $L$ is $\{x_1, \ldots, x_m\}$ then $x = \bigvee\limits_{i \in S} x_i$ and $y = \bigvee\limits_{j \in T} x_j$. Since the join-operation $\vee$ is associative, it follows that $x \vee y = \bigvee\limits_{k \in S \cup T} x_k = \sup (S \cup T)$.
\end{proof}
\begin{lemma}
	\label{UAL2}
	Suppose $L$ is an uniquely atomistic lattice. For any distinct $x, y \in L$ we must have
	\begin{equation*}
		x \wedge y = \sup (S_1 \cap S_2),
	\end{equation*}
	where $x = \sup S_1$ and $y = \sup S_2$.
\end{lemma}
\begin{proof}
	By definition of a lattice, $x \vee (x \wedge y) = x$. We can write $x \wedge y = \sup S_3$ for some finite set $S_3$ as $L$ is uniquely atomistic. That $x = \sup S_1$ and $y = \sup S_2$ implies that $\bigvee\limits_{i \in S_1 \cup S_3} x_i = \bigvee\limits_{j \in S_1} x_j$ by Lemma~\ref{UAL1}. By unique atomisticity, we get $S_3 \subset S_1$ since $x \ne x \wedge y$. Similarly, $S_3 \subset S_2$. Thus $S_3 \subseteq S_1 \cap S_2$.
	
	Assume that $S_3 \ne S_1 \cap S_2$, i.e. $S_3 \subset S_1 \cap S_2$. But that means $\sup (S_1 \cap S_2)$ is a lower bound of $x$ and $y$ and $x \wedge y \prec \sup (S_1 \cap S_2)$, a contradiction. Thus the statement follows.
\end{proof}
\begin{lemma}
	\label{UAL3}
	Suppose $x = \sup S$ and $y = \sup T$ are elements of a uniquely atomistic lattice $L$. If $x \prec y$ then $S \subset T$.
\end{lemma}
\begin{proof}
	As $x \prec y$, we have $x \wedge y = x$. From unique atomisticity of $L$ and Lemma~\ref{UAL2} it can be observed that $S \cap T = S$. The rest follows because $S \ne T$.
\end{proof}
We will now establish that modularity is inherent in any finite uniquely atomistic lattice.
\begin{theorem}
	\label{UAT}
	A finite lattice $L$ is modular if $L$ is uniquely atomistic.
\end{theorem}
\begin{proof}
	According to Theorem~\ref{TL1} it is enough to show that there exists no sublattice of $L$ which is isomorphic to $N_5$. Let the set of all atoms in $L$ be $\{x_1, \ldots, x_m\}$. We proceed by contradiction.
	\begin{figure}[t]
		\centering
		\begin{tikzpicture}[scale=0.6]
			\node (A1) at (0,2) {$\mathcal{M}$};
			\node (A2) at (-2,1) {$a_1$};
			\node (A3) at (-2,-1) {$a_2$};
			\node (A4) at (2,0) {$b$};
			\node (A5) at (0,-2) {$\mu$};
			\draw (A5) -- (A3) -- (A2) -- (A1) -- (A4) -- (A5);
		\end{tikzpicture}
		\caption{$N_5$ lattice}
		\label{F0}
	\end{figure}
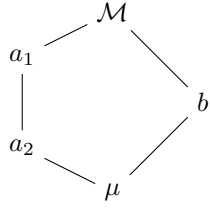		
	
	Assume that there exists a sublattice of $L$ isomorphic to $N_5$ as shown in Fig.~\ref{F0}. By unique atomisticity of $L$, we can write the following for some fixed subsets $\mathcal{I}_1, \mathcal{I}_2, \mathcal{K} \subseteq [m]$:
	\begin{equation*}
		a_1 = \bigvee\limits_{i \in \mathcal{I}_1} x_i = \sup \mathcal{I}_1; \quad a_2 = \bigvee\limits_{j \in \mathcal{I}_2} x_j = \sup \mathcal{I}_2; \quad b = \bigvee\limits_{k \in \mathcal{K}} x_k = \sup \mathcal{K}.
	\end{equation*}
	Since $\mathcal{M} = a_1 \vee b$ and $\mu = a_2 \wedge b$, we obtain from Lemmas~\ref{UAL1} and \ref{UAL2} that $\mathcal{M} = \sup (\mathcal{I}_1 \cup \mathcal{K})$ and $\mu = \sup (\mathcal{I}_2 \cap \mathcal{K})$. Similarly, from $\mathcal{M} = a_2 \vee b$ and $\mu = a_1 \wedge b$ we yield $\mathcal{M} = \sup (\mathcal{I}_2 \cup \mathcal{K}), \mu = \sup (\mathcal{I}_1 \cap \mathcal{K})$. Combining both, we have the following:
	\begin{eqnarray}
		\mathcal{I}_1 \cup \mathcal{K} &=& \mathcal{I}_2 \cup \mathcal{K}; \label{UAE1} \\
		\mathcal{I}_1 \cap \mathcal{K} &=& \mathcal{I}_2 \cap \mathcal{K}. \label{UAE2}
	\end{eqnarray}
	From \eqref{UAE2} it follows that $(\mathcal{I}_1 \backslash \mathcal{I}_2) \cap \mathcal{K} = \phi$. Since $a_2 \prec a_1$, thus $\mathcal{I}_2 \subset \mathcal{I}_1$ by Lemma~\ref{UAL3}; i.e. $\mathcal{I}_1 \backslash \mathcal{I}_2$ is nonempty. For any $l \in \mathcal{I}_1 \backslash \mathcal{I}_2$ we must have $l \in \mathcal{I}_2 \cup \mathcal{K}$ from \eqref{UAE1}, i.e. $l \in \mathcal{K}$. This implies that $\mathcal{I}_1 \backslash \mathcal{I}_2 \subseteq \mathcal{K}$, or in other words $\mathcal{I}_1 \backslash \mathcal{I}_2 = (\mathcal{I}_1 \backslash \mathcal{I}_2) \cap \mathcal{K} = \phi$. This is a contradiction to $\mathcal{I}_2 \subset \mathcal{I}_1$, hence proved.
\end{proof}
\begin{corollary}
	\label{UAC}
	Any finite uniquely atomistic lattice is geometric.
\end{corollary}
\begin{proof}
	Follows directly from Definition~\ref{LD6} and Theorem~\ref{UAT}.
\end{proof}
\begin{remark}
	The converse of the statement of Theorem~\ref{UAT} is not true, i.e. a modular lattice need not always be uniquely atomistic. E.g. the $M_3$ lattice is not uniquely atomistic. From Fig.~\ref{F1} one can see that $\mathbb{F}_2^2 = \langle \{0, 1\}\rangle + \langle \{1, 0\}\rangle = \langle \{0, 1\}\rangle + \langle \{1, 1\}\rangle$, where $+$ denotes the usual vector space addition.
\end{remark}
Theorem~\ref{UAT} brings us to a position where we can prove the unique-decomposition theorem in the next section.

\section{The Unique-decomposition Theorem}
\label{S3}
In this section we will establish the unique criterion needed for an atomistic lattice to be distributive and use that to characterize finite geometric distributive lattices. Atoms of geometric distributive lattices play an important part in the unique characterization. The first step towards that is observing that the greatest lower bound of two atoms in any lattice is the least element of that lattice.
\begin{lemma}
	\label{L1}
	For any two distinct atoms $x_1, x_2$ in a lattice $(L, \vee, \wedge)$, we have $x_1 \wedge x_2 = O$, where $O$ is the least element of $L$.
\end{lemma}
\begin{proof}
	Suppose $x_1 \wedge x_2 \neq O$. By definition, $x_1 \wedge x_2 \preceq x_1$. As $x_1$ is an atom in $L$, hence $O \lessdot x_1$, which means $x_1 \wedge x_2 = x_1$ by our supposition. Similarly we obtain $x_1 \wedge x_2 = x_2$. Since $x_1, x_2$ are distinct, this is a contradiction and the result follows.
\end{proof}

Next we will prove the generalization of the above lemma for any finite number of atoms in a distributive lattice.
\begin{lemma}
	\label{L2}
	Let $\{x_1, \ldots, x_m\}$ be a set of atoms in a distributive lattice $M$. Then for all $i \in [m]$,
	\begin{equation*}
		x_i \wedge (\bigvee\limits_{j \in [m] \backslash \{i\}} x_j) = O.
	\end{equation*}
\end{lemma}
\begin{proof}
	By distributivity in $M$ we can write,
	\begin{equation*}
		x_i \wedge (\bigvee\limits_{j \in [m] \backslash \{i\}} x_j) = \bigvee\limits_{j \in [m] \backslash \{i\}} (x_i \wedge x_j).
	\end{equation*}
	According to Lemma~\ref{L1}, $x_i \wedge x_j = O$ for $i \ne j$, and the statement is proved.
\end{proof}

The height of join of two atoms in a modular lattice (if the height function is defined) is the sum of heights of the individual atoms, as illustrated in the following lemma.
\begin{lemma}
	\label{L3}
	Suppose $L$ is a modular lattice and $h$ is the height function defined on $L$. For any two atoms $x$ and $y$ in $L$ the following holds true:
	\begin{equation*}
		h(x \vee y) = h(x) + h(y).
	\end{equation*}
\end{lemma}
\begin{proof}
	The height function $h$ is a valuation which according to Definition~\ref{LD8} implies that $h(x \vee y) = h(x) + h(y) - h(x \wedge y)$. As $x$ and $y$ are atoms in $L$, Lemma~\ref{L1} dictates that $x \wedge y = O$. That $h(O) = h(x \wedge y) = 0$ which proves the rest.
\end{proof}

We are now going to generalize Lemma~\ref{L3} for any $m \ge 2$ number of atoms if the lattice is also distributive.
\begin{lemma}
	\label{L3a}
	Consider a set $\{x_1, \ldots, x_m\}$ of $m \ge 2$ atoms in a distributive lattice $L$. If $h$ is the height function defined on $L$ then
	\begin{equation*}
		h(\bigvee\limits_{i \in [m]} x_i) = \sum\limits_{i \in [m]}h(x_i).
	\end{equation*}
\end{lemma}
\begin{proof}
	The proof is by induction. The base case for $m = 2$ is covered by Lemma~\ref{L3}. Suppose the statement holds true for any $(m-1)$ atoms in $L$, i.e., $h(\bigvee\limits_{i \in [m-1]} x_i) = \sum\limits_{i \in [m-1]}h(x_i)$. Since the join operation $\vee$ is associative over the elements of $L$, we can write $\bigvee\limits_{i \in [m]} x_i = x_m \vee (\bigvee\limits_{j \in [m-1]} x_j)$. The height of $\bigvee\limits_{i \in [m]} x_i$ can therefore be expressed as:
	\begin{equation*}
		h(\bigvee\limits_{i \in [m]} x_i) = h(x_m) + h(\bigvee\limits_{i \in [m-1]} x_i) - h(x_m \wedge (\bigvee\limits_{i \in [m-1]} x_i)).
	\end{equation*}
	As $h(x_m \wedge (\bigvee\limits_{i \in [m-1]} x_i)) = 0$ according to Lemma~\ref{L2}, the rest follows from the induction hypothesis.
\end{proof}

Consequence of Lemma~\ref{L3a} is that the number of atoms in a distributive sublattice of a finite modular lattice of height $n$ cannot exceed $n$. We formally state the result.
\begin{proposition}
	\label{P2}
	If $L$ is a finite modular lattice of height $n$ then the number of atoms in a distributive sublattice $M$ of $L$ can be at most $n$. The maximum number of atoms is reached if and only if all atoms of $M$ are also atoms in $L$ and the greatest element of $L$ is join of the atoms in $M$.
\end{proposition}
\begin{proof}
	Suppose $\{x_1, \ldots, x_m\}$ be the set of atoms in $M$. If $I$ is the greatest element in $L$ then certainly $\bigvee\limits_{i \in [m]} x_i \preceq I$. As $h$ is an isotone valuation, Definition~\ref{LD8}(ii) implies that $h_L(\bigvee\limits_{i \in [m]} x_i) \le h_L(I)$. Observe that $h_L(y) \ge h_M(y)$ for all $y \in M$. Since $h_L(I) = n$, by Lemma~\ref{L3a} we have the following inequality that proves the claim:
	\begin{equation}
		\label{E0}
		m = \sum\limits_{i \in [m]} h_M(x_i) = h_M(\bigvee\limits_{i \in [m]} x_i) \le h_L(\bigvee\limits_{i \in [m]} x_i) \le h_L(I) = n.
	\end{equation}
	Since $h$ is positive isotone, it is evident from \eqref{E0} that $m = n$ if and only if $I = \bigvee\limits_{i \in [m]} x_i$ and $h_M(\bigvee\limits_{i \in [m]} x_i) = h_L(\bigvee\limits_{i \in [m]} x_i)$. That $x_i$'s are atoms in $L$ for all $i \in [m]$ if and only if $x_1 \prec x_1 \vee x_2 \prec \cdots \prec \bigvee\limits_{j \in [m-1]} x_j \prec \bigvee\limits_{i \in [m]} x_i$ is a maximal chain in $L$ proves the rest.
\end{proof}

Any element in a geometric lattice can be expressed as a join of atoms (Definition~\ref{LD6}). However, such representation is not unique. To elaborate, if $\{x_1, \ldots, x_m\}$ is the set of atoms in a geometric lattice $L$ then there may exist $y \in L$ such that $y = \bigvee\limits_{i \in \mathcal{I} \subseteq [m]} x_i = \bigvee\limits_{j \in \mathcal{J} \subseteq [m]} x_j$ for two different subsets $\mathcal{I}, \mathcal{J} \subseteq [m]$. We will now show that representation of elements as join of atoms is unique if and only if the lattice is also distributive. In fact we will prove the following which is a more generalized statement.
\begin{theorem}[Unique-decomposition Theorem]
	\label{T1}
	A finite atomistic lattice is distributive if and only if it is uniquely atomistic.
\end{theorem}
\begin{proof}
	Let $M$ be an atomistic lattice with $\{x_1, \ldots, x_m\}$ as its set of all atoms. First we prove that $M$ is uniquely atomistic, i.e. $\bigvee\limits_{i \in \mathcal{I} \subseteq [m]} x_i \in M$ is uniquely determined by $\mathcal{I} \subseteq [m]$ when $M$ is distributive.
	
	The proof is by contradiction. Suppose the claim is false, i.e. there exist subsets $\mathcal{I}, \mathcal{J} \subseteq [m]$ such that $\mathcal{I} \ne \mathcal{J}$ and $\bigvee\limits_{i \in \mathcal{I}} x_i = \bigvee\limits_{j \in \mathcal{J}} x_j$. Since $\mathcal{I}$ and $\mathcal{J}$ are distinct, at least one of them is not contained within the other. Without loss of generality, suppose $\mathcal{I} \nsubseteq \mathcal{J}$. Then the difference set $\mathcal{I} \backslash \mathcal{J}$ is nonempty, i.e. there exists an integer $l \in \mathcal{I} \backslash \mathcal{J}$. Thus we can write as per our supposition:
	\begin{equation}
		\label{E1}
		x_l \wedge (\bigvee\limits_{i \in \mathcal{I}} x_i) = x_l \wedge (\bigvee\limits_{j \in \mathcal{J}} x_j).
	\end{equation}
	
	The individual terms can be decomposed further. Since $l \in \mathcal{I}$, by distributivity in $M$ we obtain $x_l \wedge (\bigvee\limits_{i \in \mathcal{I}} x_i) = (x_l \wedge (\bigvee\limits_{i \in \mathcal{I} \backslash \{l\}} x_i)) \bigvee (x_l \wedge x_l) = O \vee x_l = x_l$. The penultimate step follows from Lemma~\ref{L2}. Similar technique yields $x_l \wedge (\bigvee\limits_{j \in \mathcal{J}} x_j) = O$ as $l \notin \mathcal{J}$. Hence \eqref{E1} suggests $x_l = O$. However, this is a contradiction since $l \in \mathcal{I} \subseteq [m]$, i.e. $x_l$ is an atom. We conclude that our initial assumption was wrong and the representation $\bigvee\limits_{i \in \mathcal{I}} x_i$ is uniquely determined by $\mathcal{I}$.
	
	Now it remains to prove that $M$ is distributive if it is uniquely atomistic. Once again we proceed by contradiction. That $M$ is modular follows at once from Theorem~\ref{UAT}. Suppose $M$ is modular non-distributive. By Theorem~\ref{TL1} $M$ must contain a sublattice isomorphic to $M_3$. In other words there exist $y_1, y_2, y_3 \in M$ such that $y_1 \vee y_2 = y_2 \vee y_3 = y_3 \vee y_1$ and $y_1 \wedge y_2 = y_2 \wedge y_3 = y_3 \wedge y_1$, where $y_i \npreceq y_j$ for $i \neq j$ (See Fig.~\ref{F2}). Suppose $\mathcal{M} = y_1 \vee y_2$ and $\mathfrak{m} = y_1 \wedge y_2$. By imposition of unique atmosticity, $y_j = \bigvee\limits_{i \in \mathcal{J}_j} x_i$ for $j = 1, 2, 3$, where $\mathcal{J}_j \subseteq [m]$ uniquely determines $y_j$. This implies the following:
	\begin{equation}
		\label{E1a}
		\mathcal{M} = \bigvee\limits_{i \in \mathcal{J}_1 \cup \mathcal{J}_2} x_i = \bigvee\limits_{j \in \mathcal{J}_2 \cup \mathcal{J}_3} x_j = \bigvee\limits_{k \in \mathcal{J}_3 \cup \mathcal{J}_1} x_k.
	\end{equation}
	Since $M$ is uniquely atomistic, applying Lemma~\ref{UAL1} to \eqref{E1a} implies that $\mathcal{J}_1 \cup \mathcal{J}_2 = \mathcal{J}_2 \cup \mathcal{J}_3 = \mathcal{J}_3 \cup \mathcal{J}_1$.
	
	\begin{figure}
		\centering
		\begin{tikzpicture}[scale=0.7]
			\node (A1) at (0,2) {$\mathcal{M}$};
			\node (A2) at (-2,0) {$y_1$};
			\node (A3) at (0,0) {$y_2$};
			\node (A4) at (2,0) {$y_3$};
			\node (A5) at (0,-2) {$\mathfrak{m}$};
			\node[right=0pt of A1,inner xsep=0pt] {$= y_1 \vee y_2 = y_2 \vee y_3 = y_3 \vee y_1$};
			\node[right=0pt of A5,inner xsep=0pt] {$= y_1 \wedge y_2 = y_2 \wedge y_3 = y_3 \wedge y_1$};
			\draw (A5) -- (A2) -- (A1) -- (A3) -- (A5) -- (A4) -- (A1);
		\end{tikzpicture}
		\caption{$M_3$-sublattice in $M$}
		\label{F2}
	\end{figure}
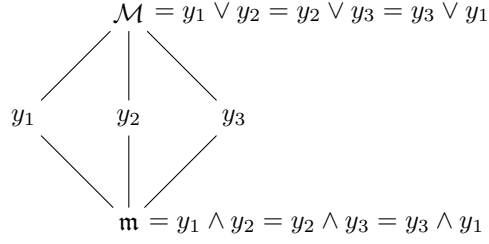
	
	On the other hand $\mathfrak{m} = y_1 \wedge y_2$, where $y_1 = \sup \mathcal{J}_1$ and $y_2 = \sup \mathcal{J}_2$. If $\mathcal{L} \subseteq [m]$ uniquely determines $\mathfrak{m}$, i.e. $\mathfrak{m} = \sup \mathcal{L}$, then Lemma~\ref{UAL2} implies that $\mathcal{L} = \mathcal{J}_1 \cap \mathcal{J}_2$. Similarly we can deduce for $y_2 \wedge y_3$ and $y_3 \wedge y_1$ which indicates that $\mathcal{J}_1 \cap \mathcal{J}_2 = \mathcal{J}_2 \cap \mathcal{J}_3 = \mathcal{J}_3 \cap \mathcal{J}_1$.
	
	We can now express $\mathcal{J}_1$ as $\mathcal{J}_1 = \mathcal{J}_1 \cap (\mathcal{J}_1 \cup \mathcal{J}_2) = \mathcal{J}_1 \cap (\mathcal{J}_2 \cup \mathcal{J}_3) = (\mathcal{J}_1 \cap \mathcal{J}_2) \cup (\mathcal{J}_1 \cap \mathcal{J}_3) = \mathcal{J}_1 \cap \mathcal{J}_2$, i.e. $\mathcal{J}_1 \subseteq \mathcal{J}_2$. This means $y_1 \preceq y_2$, which contradicts our initial assumption. Hence, $M$ is distributive.
\end{proof}

\begin{corollary}
	\label{UDT}
	A geometric lattice is distributive if and only if it is uniquely atomistic.
\end{corollary}
\begin{proof}
	A geometric lattice is finite atomistic by definition, which concludes the proof.
\end{proof}
\begin{remark}
	The semimodularity of a geometric lattice is not required for it to be distributive.
\end{remark}
The consequence of Corollary~\ref{UDT} is that any element of a geometric distributive lattice can be uniquely decomposed as join of its atoms. The next statement also follows from Theorem~\ref{T1}:
\begin{corollary}
	\label{P3}
	The greatest element in a geometric distributive lattice is the join of all of its atoms.
\end{corollary}
\begin{proof}
	Let $\{x_1, \ldots, x_m\}$ be the set of all atoms in a geometric distributive lattice $M$. We aim to show that the greatest element in $M$ is $g := \bigvee\limits_{i \in [m]} x_i$. To that end, say $y \in M$ is the greatest element in $M$. By Theorem~\ref{T1} we can write $y = \bigvee\limits_{i \in \mathcal{I}} x_i$ for some $\mathcal{I} \subseteq [m]$. However, by associativity of the join operation $\vee$ in $M$, $g$ can also be decomposed as $g = (\bigvee\limits_{j \in [m] \backslash \mathcal{I}} x_j) \vee (\bigvee\limits_{i \in \mathcal{I}} x_i) = (\bigvee\limits_{j \in [m] \backslash \mathcal{I}} x_j) \vee y$, which implies $y \preceq g$; thus the only possibility is $y = g$, which concludes the proof.
\end{proof}
In the following section we will see a few applications of the Unique-decomposition theorem, mainly in the context of linearity and complements in $\mathbb{P}_q(n)$.

\section{Applications of the Unique-decomposition Theorem}
\label{S4}
The unique-decomposition theorem, akin to unique decomposition in context of linear subspace codes \cite[Proposition~12]{BK2}, lays the path for determining the maximum size of a distributive sublattice of a finite geometric lattice. We also characterize the extremal case.
\begin{theorem}
	\label{T2}
	The size of any distributive sublattice of a finite geometric lattice of height $n$ can be at most $2^n$. The bound is reached if and only if each of the atoms in the sublattice is also an atom in the geometric lattice and the greatest element of the geometric lattice belongs to the distributive sublattice.
\end{theorem}
\begin{proof}
	Suppose $M$ is a distributive sublattice of a finite geometric lattice $L$ with the height function $h$ defined on $L$ such that $h(I) = n$, where $I$ is the greatest element of $L$. We require to show that $|M| \le 2^n$.
	
	By supposition $M$ is geometric. Let $\{x_1, \ldots, x_m\}$ be the set of all atoms in $M$. We define a mapping $\Phi$ from $\mathcal{P}(m)$ to $M$ as below:
	\begin{eqnarray*}
		\Phi : \mathcal{P}(m) &\longrightarrow& M \nonumber \\
		\mathcal{I} &\mapsto& \bigvee\limits_{i \in \mathcal{I}} x_i.
	\end{eqnarray*}
	By definition the map $\Phi$ is well-defined. Suppose there exist $\mathcal{I}, \mathcal{J} \in \mathcal{P}(m)$ such that $\Phi(\mathcal{I}) = \Phi(\mathcal{J})$, i.e., $\bigvee\limits_{i \in \mathcal{I}} x_i = \bigvee\limits_{j \in \mathcal{J}} x_j$. As $M$ is geometric distributive, Corollary~\ref{UDT} dictates that $\mathcal{I} = \mathcal{J}$; thus $\Phi$ is injective. To check that $\Phi$ is also surjective, any $y \in M$ can be expressed as $y = \bigvee\limits_{i \in \mathcal{S} \subseteq [m]} x_i$ for some $\mathcal{S} \subseteq [m]$ as $M$ is geometric; the choice of $\mathcal{S}$ is unique according to Corollary~\ref{UDT}, hence $y = \Phi(\mathcal{S})$. Therefore $\Phi$ is a bijective map, which implies that $|M| = |\mathcal{P}(m)| = 2^m$. Combining this with Proposition~\ref{P2} yields $|M| \le 2^n$.
	
	For the extremal case we must have $m = n$. The rest then follows from Proposition~\ref{P2}.
\end{proof}
\begin{remark}
	An atom in a sublattice $M$ of a geometric lattice $L$ may not be an atom in $L$. E.g. consider the lattice $\mathcal{P}(4)$, the set of all subsets of $\{1, 2, 3, 4\}$. It is a geometric lattice with atoms $\{1\}, \{2\}, \{3\}$ and $\{4\}$. The sublattice $M = \{\phi, \{1, 2\}, \{3, 4\}, \{1, 2, 3, 4\}\}$ of $\mathcal{P}(4)$ has atoms $\{1, 2\}$ and $\{3, 4\}$, none of which is an atom in $\mathcal{P}(4)$. On the other hand, a proper sublattice of a geometric lattice $L$ does not contain all atoms of $L$.
\end{remark}
\begin{remark}
	Irrespective of the choice of the lattice $M$, the mapping $\Phi$ in Theorem~\ref{T2} always maps the ground set $[m]$ to $I$ and the empty subset $\phi$ to $O$.
\end{remark}
Class of distributive sublattices of maximum size in a finite geometric lattice can be characterized in an alternative way.
\begin{corollary}
	\label{C1}
	The size of a distributive sublattice $M$ of a finite geometric lattice $L$ of height $n$ is $2^n$ if and only if $M$ contains $n$ atoms of $L$.
\end{corollary}
\begin{proof}
	By Proposition~\ref{P2} $M$ can have at most $n$ atoms, and for the extremal case all of them belong to $L$. Since $M$ is geometric distributive, proof technique of Theorem~\ref{T2} suggests that $|M| = |\mathcal{P}(n)| = 2^n$.
	
	Conversely, if $|M| = 2^n$ then according to Theorem~\ref{T2} each atom of $M$ is also an atom in $L$ and $I \in M$. Suppose $\{x_1, \ldots, x_m\}$ is the set of all atoms in $M$; thus $h_L(x_i) = 1$ for all $i \in [m]$. By Corollary~\ref{P3}, $\bigvee\limits_{i \in [m]} x_i$ is the greatest element in $M$, which implies that $I = \bigvee\limits_{i \in [m]} x_i$. By Lemma~\ref{L3a}, $m = \sum\limits_{i \in [m]}h_L(x_i) = h_L(\bigvee\limits_{i \in [m]} x_i) = n$, which settles the proof.
\end{proof}
\begin{proof}[Proof of Theorem~\ref{LT3}]
	The upper bound follows at once from Theorem~\ref{LT2} and Theorem~\ref{T2}. Corollary~\ref{C1} implies that the maximal case occurs if and only if the number of one-dimensional codewords is $n$, i.e. the code is derived from a fixed basis.
\end{proof}
We now consider the Whitney numbers of a geometric distributive lattice.
\begin{corollary}
	\label{C2}
	The Whitney numbers of a distributive sublattice $L$ of a finite geometric lattice of height $n$ are bounded by $W_k(L) \le \binom{n}{k}$ for all $k \in \{0, 1, \ldots, n\}$. Equality occurs if and only if $L$ contains $n$ atoms.
\end{corollary}
\begin{proof}
	Let the set of atoms in $L$ be $\{x_1, \ldots, x_m\}$. The case of $k = 0$ is obvious since $h(O) = 0$. As  shown in proof of Theorem~\ref{T2}, there exists a bijection from $L$ to $\mathcal{P}(m)$ that sends $\bigvee_{i \in \mathcal{I}} x_i$ to $\mathcal{I} \subseteq [m]$. From Lemma~\ref{L3a} it follows that $W_k(L) = \binom{m}{k}$ for all $k \in [m]$. As $m \le n$ by Proposition~\ref{P2}, the result follows. The bound is achieved if and only if $m = n$.
\end{proof}
\begin{remark}
	The statement in Corollary~\ref{C2} is similar in nature to the main result in Pai and Rajan's paper \cite[Theorem~2]{PS}.
\end{remark}
In the sequel we will consider the linear lattice $\mathbb{P}_q(n)$ instead of a geometric lattice in general. $\mathbb{P}_q(n)$ is non-distributive geometric. There remain a few unanswered questions regarding linearity and complements in $\mathbb{P}_q(n)$ that can be resolved by applying the unique-decomposition theorem. It was shown before that a linaer code in $\mathbb{P}_q(n)$ that is closed under intersection necessarily is a distributive sublattice of the geometric lattice $\mathbb{P}_q(n)$ \cite{BK}. Using the unique-decomposition theorem we now prove the converse.
\begin{theorem}
	\label{T3}
	A subset $\mathcal{U} \subseteq \mathbb{P}_q(n)$ is a distributive sublattice of the corresponding linear lattice $\mathbb{P}_q(n)$ if and only if $\mathcal{U}$ is a linear code closed under intersection.
\end{theorem}
\begin{proof}
	Suppose $\mathcal{U}$ is a distributive sublattice of $\mathbb{P}_q(n)$ and $\{X_1, \ldots, X_m\}$ is the set of all atoms in $\mathcal{U}$. Obviously $\mathcal{U}$ is geometric distributive, which according to Lemma~\ref{L2} implies that
	\begin{equation}
		\label{E3}
		X_i \cap (\sum_{j \in [m] \backslash \{i\}} X_j) = \{0\}, \qquad \forall i \in [m].
	\end{equation}
	Applying the unique-decomposition theorem it is easy to see that any $Y \in \mathcal{U}$ can be uniquely expressed as $Y = \sum\limits_{i \in \mathcal{I}} X_i$ for some fixed $\mathcal{I} \subseteq [m]$. If we choose arbitrary bases $B_i$ that span $X_i$ over $\mathbb{F}_q$ for all $i \in [m]$ then \eqref{E3} implies that the set $\{B_1, \ldots, B_m\}$ is a partition of $B := \bigcup_{i=1}^{m} B_i$, a linearly independent subset of $\mathbb{F}_q^n$ over $\mathbb{F}_q$. Expressing any $Y = \sum\limits_{i \in \mathcal{I}} X_i$ as $Y = \langle B_{\mathcal{I}}\rangle$ where $B_{\mathcal{I}} := \cup_{i \in \mathcal{I}} B_i$, we can say by Theorem~\ref{LT1} that $\mathcal{U} = \{\langle B_{\mathcal{I}}\rangle: \mathcal{I} \subseteq [m]\}$ is a linear code closed under intersection with linear addition $\boxplus$ of two codewords $Y_1 = \sum\limits_{j \in \mathcal{I}_1} X_j$ and $Y_2 = \sum\limits_{l \in \mathcal{I}_2} X_l$ defined as:
	\begin{equation*}
		Y_1 \boxplus Y_2 := \langle B_{\mathcal{I}_1 \triangle \mathcal{I}_2}\rangle = \sum\limits_{j \in \mathcal{I}_1 \triangle \mathcal{I}_2} X_j.
	\end{equation*}
	
	The converse is basically the statement of Theorem~\ref{LT2}.
\end{proof}
The maximum size of a subset of $\mathbb{P}_q(n)$ wherein a complement function can be defined is hitherto unknown. We next investigate any such subset of $\mathbb{P}_q(n)$ that has a distributive sublattice structure.
\begin{theorem}
	\label{T4}
	A subset $\mathcal{U} \subseteq \mathbb{P}_q(n)$ is a distributive sublattice of the linear lattice $\mathbb{P}_q(n)$ with a complement function defined on $\mathcal{U}$ if and only if $\mathcal{U}$ is a linear code closed under intersection with $\mathbb{F}_q^n \in \mathcal{U}$.
\end{theorem}
\begin{proof}
	Suppose $\mathcal{U}$ is a distributive sublattice of $\mathbb{P}_q(n)$ and $f : \mathcal{U} \rightarrow \mathcal{U}$ is a complement on $\mathcal{U}$. It follows directly from Theorem~\ref{T3} that $\mathcal{U}$ is a linear code closed under intersection. For any $X \in \mathcal{U}$, the direct sum of $X$ and $f(X)$ is $X \oplus f(X) = \mathbb{F}_q^n$. Since $X \cap f(X) = \{0\}$ by definition of $f$, it follows from Lemma~\ref{LL1} that $X \boxplus f(X) = X \oplus f(X) = \mathbb{F}_q^n \in \mathcal{U}$.
	
	Conversely, if $\mathcal{U}$ is a linear code closed under intersection with $\mathbb{F}_q^n \in \mathcal{U}$ then the function $f$ defined as $f(X) := X \boxplus \mathbb{F}_q^n$ for all $X \in \mathcal{U}$ serves as a complement function on $\mathcal{U}$. $\mathcal{U}$ is a distributive sublattice of $\mathbb{P}_q(n)$ according to Theorem~\ref{LT2}.
\end{proof}
\begin{corollary}
	\label{C3}
	The largest distributive sublattice of $\mathbb{P}_q(n)$ on which a complement function can be defined is a code derived from a fixed basis.
\end{corollary}
\begin{proof}
	By Theorem~\ref{T4} a distributive sublattice of $\mathbb{P}_q(n)$ on which a complement function can be defined is a linear code closed under intersection containing $\mathbb{F}_q^n$. Theorem~\ref{LT3} indicates that the code has a maximum size if and only if it is derived from a fixed basis.
\end{proof}
The maximum size of a sublattice of $\mathbb{P}_q(n)$ wherein a complement function can be defined is, however, unknown. It needs to be investigated first whether a complement function can exist in a non-distributive sublattice of $\mathbb{P}_q(n)$.
\section{Conclusion}
\label{S5}
In this paper we have derived the unique criterion required for an atomistic lattice to be distributive and used that to characterize all finite geometric distributive lattices. Using the unique characterization we were able to prove that the size of a distributive sublattice of a finite geometric lattice of height $n$ is always upper bounded by $2^n$, which was conjectured in \cite{PS}. We also applied the characterization to the linear lattice $\mathbb{P}_q(n)$, which is geometric, to obtain certain results regarding linearity and complements in $\mathbb{P}_q(n)$.

Theorem~\ref{T3} states that any distributive sublattice of $\mathbb{P}_q(n)$ can be used to construct a linear code in $\mathbb{P}_q(n)$ that is closed under intersection. This result is similar to the fact that certain $d$-intersecting families in $\mathbb{G}_q(n, 2d)$ can always be used to construct \emph{equidistant} linear codes with constant distance $2d$ \cite{PB}. It might be interesting to find a more generalized statement that encapsulates the essence of both these results.

\section*{Acknowledgement}

The author would like to thank Prof. Navin Kashyap and Dr. Arijit Ghosh for their insightful comments.

\bibliographystyle{ieeetr}
\bibliography{CFGDL}

\end{document}